\documentclass[a4paper,UKenglish,pdfa]{lipics-v2021}

\nolinenumbers
\hideLIPIcs  

\usepackage[utf8]{inputenc} 
\usepackage[T1]{fontenc}    
\usepackage{hyperref}
\usepackage{booktabs}       

\usepackage{graphicx} 
\usepackage{pgfplots} 
\pgfplotsset{compat=1.18}

\usepackage{tikz} 
\usetikzlibrary{matrix}
\usetikzlibrary{shapes,arrows}
\usetikzlibrary{trees,positioning,fit}
\tikzstyle{cloud} =
[draw, text centered, ellipse, minimum height=2.0em, minimum width=9em]
\usetikzlibrary{patterns}

\usepackage{float}
\usepackage{xspace}
\usepackage{amsfonts}       
\usepackage{nicefrac}       
\usepackage{amsmath}
\usepackage{amssymb}
\usepackage{mathtools}

\def\cR{\mathcal{R}}
\def\cN{\mathcal{N}}
\def\cA{\mathcal{A}}

\newcommand{\patternedcircle}[1]{ 
	\begin{tikzpicture}
		\draw[thick, color=#1] (0,0) circle (0.1cm);
	\end{tikzpicture}
}

\newcommand{\crossedfilledcircle}[1]{ 
	\begin{tikzpicture}
		\draw[thick, color=#1] (0,0) circle (0.1cm);
   	 \fill[#1] (0,0) circle (0.04cm); 
	\end{tikzpicture}
}

\newcommand{\filledcircle}[1]{ 
	\begin{tikzpicture}
		\fill[#1] (0,0) circle (0.1cm);
	\end{tikzpicture}
}

\newcommand{\blub}{\filledcircle{blue}}
\newcommand{\redr}{\filledcircle{red}}
\newcommand{\trat}{\patternedcircle{black}}
\newcommand{\grag}{\crossedfilledcircle{gray}}

\def\minisat{\textsf{MiniSat}\xspace}
\def\nauty{\textsf{nauty}\xspace}
\def\holfour{\textsf{HOL4}\xspace}

\def\hollight{\textsf{HOL Light}\xspace}

\def\coq{\textsf{Coq}\xspace}

\title{A Formal Proof of R(4,5)=25} 
\titlerunning{A Formal Proof of R(4,5)=25}
\author{Thibault Gauthier}{Czech Technical 
University in Prague, Prague, Czech Republic}
{email@thibaultgauthier.fr}{[https://orcid.org/0000-0002-7348-0602]}{}
\author{Chad E. Brown}{Czech Technical University in Prague, Prague, Czech 
Republic}{}{}{}
\funding{Both authors were supported by the Ministry of Education, Youth and 
Sports within the dedicated program ERC CZ under the project POSTMAN no. LL1902.
The second author has also received funding from the European Union’s Horizon
Europe research and innovation programme under grant agreement
no.~101070254 CORESENSE.}
\authorrunning{T. Gauthier, Chad E. Brown}
\Copyright{Thibault Gauthier and Chad E. Brown} 
\relatedversion{}
\relatedversiondetails{Full Version}{https://arxiv.org/abs/2404.01761}

\EventEditors{Yves Bertot, Temur Kutsia, and Michael Norrish}
\EventNoEds{3}
\EventLongTitle{15th International Conference on Interactive Theorem Proving (ITP 
2024)}
\EventShortTitle{ITP 2024}
\EventAcronym{ITP}
\EventYear{2024}
\EventDate{September 9-14, 2024}
\EventLocation{Tbilisi, Georgia}
\EventLogo{}
\SeriesVolume{309}
\ArticleNo{5}

\begin{document}

\maketitle

\begin{abstract}
In 1995, McKay and Radziszowski proved that the 
Ramsey number R(4,5) is equal to 25. Their proof relies on
a combination of high-level arguments and computational steps.
The authors have performed the computational parts of the proof with different 
implementations in order to reduce the possibility of an error in their programs.
In this work, we prove this theorem in the interactive theorem prover HOL4 
limiting the uncertainty to the small HOL4 kernel.
Instead of verifying their algorithms directly, we rely on the HOL4 interface to 
MiniSat to prove 
gluing lemmas.
To reduce the number of such lemmas and thus make the computational part of the 
proof feasible, we implement a generalization algorithm. We verify that its output 
covers all the possible cases by implementing a custom SAT-solver extended with a 
graph isomorphism checker.
\end{abstract}

\section{Introduction}
Formalizations are useful to verify that there are no bugs in some 
software and also that there are no errors in a mathematical proof.
Researchers write their formalizations in an interactive theorem 
prover also called a proof assistant. An interactive theorem prover transforms 
high-level 
proof steps, written
by its user in the language of the interactive theorem prover, 
into low-level proof steps at the level of logical rules and axioms.
These low-level steps are then verified by the kernel of the proof assistant.
Formalizations are thus doubly appropriate when a proof combines 
advanced human-written arguments and computer-generated lemmas. 
This is the case for the four-color theorem~\cite{appel1989every} which was proved 
by 
Appel and Haken in 1976
and the Kepler conjecture~\cite{DBLP:journals/dcg/HalesF06} which was proved by 
Hales and 
Ferguson in 
1998.

In those two cases, a human argument is used to reduce a potentially
infinite number of cases to a finite number. Then, a computer algorithm is used to 
generate a proof for each of these cases. 
The generated proofs are too numerous to be verified manually and so the 
generating code, which is in those cases quite complicated, had to be trusted.
To avoid trusting that the generating code fits together with the human 
argument, a formalization of the four-color 
theorem~\cite{gonthier08ams} was completed in the \coq proof assistant~\cite{coq} 
by Gonthier in 2005 and a formalization of the Kepler 
conjecture~\cite{DBLP:journals/corr/HalesABDHHKMMNNNOPRSTTTUVZ15} was completed 
in the \hollight proof assistant~\cite{harrison09hollight} by a 
team led by Hales in 2014.

The case of $R(4,5) \leq 25$ is different. Since it is a finite problem,
one could prove it by considering a finite number of cases. Since there
are $\frac{25\times24}{2}=300$ edges in a graph with 25 vertices, a naive proof
would consist of checking the presence of a 4-clique or a 5-independent 
set in all graphs of size 25 which would amount to $2^{300}\approx 
10^{90}$ graphs. 
Another approach would be to encode the clique constraints into a SAT 
solver. This reduces the search space dramatically but so far
no proof of $R(4,5) \leq 25$ relying only on calls to SAT solvers has been found.

The proof of McKay and Radziszowski~\cite{DBLP:journals/jgt/McKayR95} first uses a 
high-level argument and then relies on a pre-processing algorithm
to reduce the number of cases to a
manageable number. Each of those cases requires proving that a pair of graph
cannot occur together in an $\cR(4,5,25)$-graph. 
These kinds of problems are called gluing problems. 
Our formalization of $R(4,5)=25$ 
in the \holfour theorem prover~\cite{hol4} 
will mostly follow the initial
splitting argument. We construct a slightly different pre-processing algorithm 
that uses gray edges instead of removing vertices. We also make use of the 
\holfour interface~\cite{DBLP:journals/japll/WeberA09} to the SAT solver 
\minisat~\cite{DBLP:conf/sat/EenS03}, instead of re-using the custom-built
solver of McKay and Radziszowski, to prove that the gluing problems are 
unsatisfiable. 
This greatly simplifies our proof as we do not need to trace
the proof steps of their optimized solver and we do not have to replay those proof 
steps in \holfour. 
Additional differences between our formal proof and
the original proof are discussed in Section~\ref{sec:rel}.
 

We now explain in more detail the different components of our formal proof.
To conclude that $R(4,5) = 25$, we prove that $R(4,5) \leq 25$
and that $R(4,5) > 24$. The statement $R(4,5) > 24$ can simply be 
proved by exhibiting an $\cR(4,5,24)$-graph. The existence of such
graph has been known since 1965 thanks to a construction by 
Kalbfleisch~\cite{kalbfleisch1965construction}.
The formal proof of the existence of an $\cR(4,5,24)$-graph is given in 
Section~\ref{sec:existence}. The other parts 
of this paper describe how to formally prove the more challenging
statement $R(4,5) \leq 25$.
In Section~\ref{sec:prelim}, we first give three important definitions. In 
particular,
we define the Ramsey number $R(4,5)$ which is necessary to state the 
final theorem.
In Section~\ref{sec:deg}, we prove that in an $\cR(4,5,25)$-graph there exists a
vertex of degree $d\in \{8, 10 ,12 \}$ and that the neighbors of
that vertex form an $\cR(3,5,d)$-graph and the antineighbors form an
$\cR(4,4,24-d)$-graph as illustrated in Figure~\ref{fig:split}.
This vertex is referred to 
in other parts of the proofs as the splitting vertex.
In Section~\ref{sec:enum}, we enumerate all possible
$\cR(3,5,d)$-graphs and $\cR(4,4,24-d)$-graphs modulo isomorphism. We then
regroup similar graphs together in what we call generalizations.
In Section~\ref{sec:gluing}, we prove that there is no way to glue
an $\cR(3,5,d)$-generalization and an $\cR(4,4,24-d)$-generalization while 
respecting the clique constraints. This is achieved by encoding the gluing into 
SAT and calling the \holfour interface to \minisat.
In Section~\ref{sec:better}, we improve the construction of generalizations by 
preferring ones resulting in easier gluing problems. This selection is
guided by a simplicity heuristic, which 
estimates how hard the resulting SAT solving problems would be, as described in  
Section~\ref{sec:simplicity}.
In Section~\ref{sec:connect}, we connect the different parts of the proofs
proving that formulas stated at different logical levels (propositional, 
first-order and higher-order) imply each other in the
desired way.

\begin{remark*}
Not every algorithm needs to have its computation steps verified in a formal 
manner. Sometimes, it is enough to verify that the mathematical object 
produced by the algorithm satisfies the desired properties. For example, we did 
not verify 
every step of the \nauty algorithm~\cite{DBLP:journals/jsc/McKayP14} which we rely 
on to normalize graphs  
in Section~\ref{sec:enum}.
Indeed, it is sufficient to save the witness permutations used during 
graph normalization to show that two graphs are isomorphic.
\end{remark*}

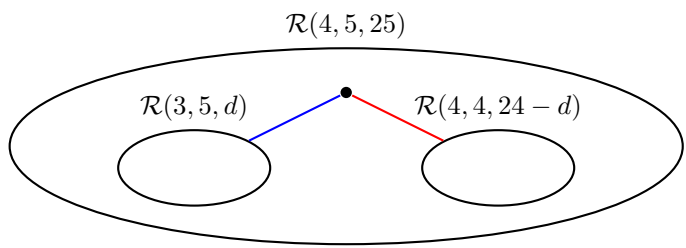
\begin{figure}
	\centering
	\begin{tikzpicture}[scale=1.0, every node/.style={scale=1.0}]
		\node[circle,thick,fill,inner sep=1.5pt] (0) {};
		\node[below of=0] (0b) {};
		\node[cloud,thick,left of=0b, node distance=2cm, minimum height=1cm, 
		minimum 
		width=2cm] 
		(1) {};
		\node[above of=1, yshift=-0.2cm] (1a) {$\cR(3,5,d)$};
		\node[cloud,thick,right of=0b, node distance=2cm, minimum height=1cm, 
		minimum 
		width=2cm] 
		(2) {};
		\node[above of=2, yshift=-0.2cm] (2a) {$\cR(4,4,24-d)$};
		\node[cloud,thick,fit=(0) (1) (2),minimum 
		width=8cm] (3) {};
		\node[above of=3,yshift=0.6cm] (3a) {$\cR(4,5,25)$};
		
		\draw[thick,blue] (0) -- (1);
		\draw[thick,red] (0) -- (2);
		
	\end{tikzpicture}
	\caption{Neighbors (blue-neighbors) and antineighbors (red-neighbors) of a 
		vertex of degree d in an $\cR(4,5,25)$-graph. \label{fig:split}}
\end{figure}

\section{Preliminaries}\label{sec:prelim}

Throughout our proof, we rely on the following definitions.

\begin{definition}(neighbors, antineighbors)\\
Given a graph $(V,E)$:
\begin{itemize}
\item the set of neighbors (blue-neighbors) of a vertex $x\in V$ is
$\lbrace y\in V\ |\ y \not= x \wedge (x,y) \in E \rbrace$,
\item the set of antineighbors (red-neighbors) of a vertex $x\in V$ 
is $\lbrace y\in V\ |\ y \not= x \wedge (x,y) \not\in E \rbrace$.
\end{itemize}
\end{definition}

\begin{definition}(Ramsey property)\\
The Ramsey property $\cR(n,m,k)$ holds for a graph $(V,E)$ if:
\begin{itemize}
\item $V$ has size $k$,
\item $(V,E)$ does not contain a clique (blue-clique) of size n,
\item $(V,E)$ does not contain an independent set (anticlique, red-clique) of 
size m.
\end{itemize}
We also use $\cR(n,m,k)$ to refer to the set of graphs
with property $\cR(n,m,k)$.\\
A graph for which the property $\cR(n,m,k)$ holds is called a 
$\cR(n,m,k)$-graph.
\end{definition}
 
\begin{definition}(Ramsey number)\\
The Ramsey number $R(n,m)$ is the least $k\in \mathbb{N}$ such that 
$\cR(n,m,k)$ is empty.
\end{definition}

In our formalization, a set of vertices $V$ will be represented by a subset of 
nonnegative integers.  Moreover, we often use an equivalent formulation 
of graphs when discussing algorithms on graphs. In the equivalent formulation,
all graphs are complete graphs but their edges are colored either blue or red.
The correspondence between the two formulations
is straightforward. There is a blue edge in the second formulation if 
and only if there is an edge in the first.

\section{Degree Constraints}\label{sec:deg}

As an intermediate concept, we define $R^o(r,s,n)$
to hold if $\cR(r,s,n)$ is empty.
It is easy to see $R(r,s) \leq n$ iff $R^o(r,s,n)$.
In our formalization, we are primarily interested in
proving $R^o(r,s,n)$ for values of $r$, $s$ and $n$.
In this section our focus is on reducing the goal
of proving $R^o(4,5,25)$ to ruling out vertices
of degrees 8, 10 or 12.
All the lemmas presented in this section are
formalized in the file \texttt{basicRamsey} of our repository~\cite{ramsey-github}.
These lemmas are reformulations of basic results
in graph theory~\cite{DBLP:books/daglib/0009415}.

Given a graph $(V,E)$\footnote{We always implicitly assume the set $V$ is finite.}
and a vertex $v\in V$, we write $\cN^{(V,E)}(v)$ for the
set of neighbors of $v$ and $\cA^{(V,E)}(v)$ for the set of antineighbors of $v$.
We will almost always omit the superscript and write $\cN(v)$ and $\cA(v)$.
The degree of $v$ is defined to be the cardinality of $\cN(v)$.
Likewise, the antidegree of $v$ is the cardinality of $\cA(v)$.


Several relevant smaller Ramsey numbers are well-known:
$R(2,s) = s$, $R(3,3) = 6$, $R(3,4) = 9$, $R(3,5) = 14$
and $R(4,4) = 18$.
In our formalization we prove the $R^o$ variant,
only proving the known values are upper bounds.
We begin by sketching a description of these results
as well as some of the preliminary results used to obtain them.

By considering complements of graphs we know
that $R^o(r,s,n)$ implies $R^o(s,r,n)$.
If $(V,E)\in\cR(r+1,s,n)$ and $v\in V$ is a vertex of degree $d$, then
$(\cN(v),E)\in\cR(r,s,d)$.
Likewise, if $(V,E)\in\cR(r,s+1,n)$ and $v\in V$ is a vertex with antidegree $d$, 
then
$(\cA(v),E)\in\cR(r,s,d)$.

Every graph in $\cR(2,s,m)$ has no edges (since an edge would be a 2-clique).
Thus every graph in $\cR(2,s,m)$ is an independent set of size $m$.
This is impossible if $s\leq m$, and so we conclude $R^o(2,m,m)$.
Likewise, $R^o(m,2,m)$.

We next prove a well-known result that provides upper bounds for values of 
$R(r,s)$.
\begin{lemma}\label{lem:ramseysum}
	If $R^o(r+1,s,m+1)$ and $R^o(r,s+1,n+1)$, then $R^o(r+1,s+1,m+n+2)$.
\end{lemma}
\begin{proof} Assume we have a graph $(V,E)$
	in $\cR(r+1,s+1,m+n+2)$. We choose a vertex $v\in V$ with degree $d$ and 
	antidegree $d'$.
	We know $(\cN(v),E)\in \cR(r,s+1,d)$
	and $(\cA(v),E)\in \cR(r+1,s,d')$.
	We obtain a contradiction using
	$d + d' = m+n+1$, $d < m+1$ (since $R^o(r+1,s,m+1)$) and $d' < n+1$ (since 
	$R^o(r,s+1,n+1)$).
\end{proof}

Applying the previous results, we immediately obtain $R^o(3,3,6)$.
We also obtain $R^o(3,4,10)$, but need the stronger result $R^o(3,4,9)$.

There is an easy informal argument for why $\cR(3,4,9)$ is empty.
Assume $(V,E)$ is a graph in $\cR(3,4,9)$.
The results above ensure every vertex $v\in V$ must have degree $d < 4$
(since $R^o(2,4,4)$) and antidegree $d' < 6$ (since $R^o(3,3,6)$).
Since $d+d'=8$, we must have $d=3$ and $d'=5$.
We now consider the sum of the degrees of each vertex.
Since the relation is symmetric, the sum must be even, as each edge
is counted as part of the degree of each of the vertices of the edge.
However, the sum is also $9\cdot 3 = 27$, which is odd.
Hence no such graph exists.
Below we describe our formalization of general results
allowing us to prove $R^o(3,4,9)$. The results will also
allow us to later prove every graph in $\cR(4,5,25)$ must
have a vertex with even degree.

\begin{lemma}\label{lem:odddegrsum} Let $(V,E)$ be a graph in which every vertex 
has odd degree.
	For each $U\subseteq V$, $U$ has odd cardinality if and only if $\Sigma_{u\in 
	U}|\cN(u)|$ is odd.
\end{lemma}
\begin{proof} The proof follows by an induction over the finite set $U$.
\end{proof}
Applying Lemma~\ref{lem:odddegrsum} with $U=V$, we obtain that if $V$ has odd
cardinality and every vertex has odd degree, then $\Sigma_{v\in V}|\cN(v)|$ is odd.
In particular for a hypothetical graph $(V,E)\in\cR(3,4,9)$,
$\Sigma_{v\in V}|\cN(v)|$ is odd since $9$ and $3$ are odd.

On the other hand we can prove $\Sigma_{v\in V}|\cN(v)|$ is always even,
though this requires two inductions on finite sets.
We first prove that if we extend a graph with a new vertex,
the neighbors of the new vertex in the larger graph contribute twice
to the sum.
\begin{lemma}\label{lem:sumdegrevenlem12} Let $V$ be a finite set, $u\notin V$ and
	$E$ be a symmetric relation (on $V\cup\{u\}$).\footnote{In the formalization, 
	$E$ is assumed to be symmetric on the relevant type, ignoring $V\cup\{u\}$.}
	For every finite set $U$, if $\cN^{(V\cup\{u\},E)}(u) = U$, then
	$$\Sigma_{w\in V\cup\{u\}} |\cN^{(V\cup\{u\},E)}(w)|
	= \Sigma_{v\in V} |\cN^{(V,E)}(v)| + 2|U|.$$
\end{lemma}
\begin{proof}
	This is proved by induction on the finite set $U$.
\end{proof}
We can now prove the sum is even by induction on the finite set of vertices $V$.
\begin{lemma}\label{lem:sumdegreven} For every finite set $V$ and symmetric 
relation $E$,
	$\Sigma_{v\in V} |\cN^{(V,E)}(v)|$ is even.
\end{lemma}

With Lemmas~\ref{lem:odddegrsum} and~\ref{lem:sumdegreven}
we can conclude $R^o(3,4,9)$ since the sum of the degrees
of the vertices in a hypothetical graph $(V,E)\in\cR(3,4,9)$
would be both odd and even.

Using Lemma~\ref{lem:ramseysum} we now immediately obtain
$R^o(3,5,14)$ and $R^o(4,4,18)$, giving us all the upper bounds
for small Ramsey numbers we will need.


We now turn to the consideration of $R^o(4,5,25)$.
For the next steps in the proof, we assume for the sake of contradiction that 
there exists a graph $(V,E) \in \cR(4,5,25)$.
Let $v\in V$ with degree $d$ and antidegree $d'$ be given.
Since $(\cN(v),E)\in\cR(3,5,d)$
and $(\cA(v),E)\in\cR(4,4,d')$
we know $d < 14$ and $d' < 18$.
Since $d+d'=24$, we must have $d > 6$.
This provides our basic upper and lower bounds
on degrees of vertices in $(V,E)$.

These same degree bounds are, of course, given 
in~\cite{DBLP:journals/jgt/McKayR95}.
The argument in~\cite{DBLP:journals/jgt/McKayR95} considers graphs in
$\cR(3,5,d)$ and corresponding graphs in
$\cR(4,4,24-d)$ that could hypothetically correspond to $\cN(v)$
and $\cA(v)$ for a vertex $v\in V$.
In \cite{DBLP:journals/jgt/McKayR95}, the case with $d=11$ is ruled out since if 
every vertex had degree $11$, the sum of degrees would be odd, giving
a contradiction.
That is, we can be assured of the existence of a vertex
$v\in V$ with degree $d\in\{7,8,9,10,12,13\}$.
In our proof, we apply Lemmas~\ref{lem:odddegrsum} and~\ref{lem:sumdegreven}
more generally to conclude that there must be a vertex
$v\in V$ of even degree.
Thus, we can be assured there is a $v\in V$ with degree $d\in\{8,10,12\}$.

\section{Enumeration of Graphs and Construction of Covers}~\label{sec:enum}

Assuming that there exists a graph $(V,E) \in \cR(4,5,25)$, there 
must exist a vertex $v \in V$ of degree $d \in 
\lbrace 8,10,12\rbrace$ as proven in Section~\ref{sec:deg}. Thus, if we prove that
for all $d \in \lbrace 8,10,12 \rbrace$ and for all pair of graphs $G\in 
\cR(3,5,d)$ and $H\in \cR(4,4,24-d)$, there is no way to color 
edges connecting $G$ and $H$ without creating a 4-blue or a 5-red clique, then we 
would have proved that $R^o(4,5,25)$ (i.e. $R(4,5) \leq 25$).

Here is a simple approach.
First, enumerate all the graphs in 
$\cR(3,5,d)$ and in 
$\cR(4,4,24-d)$, and then prove the absence of gluing between each pair of graphs 
(see Section~\ref{sec:gluing}).
This is however not efficient enough given our computational means.
In Table~\ref{tab:reduction}, we estimated that this approach would 
take more than 16,000 CPU days. 
To save time in both algorithms, we regroup graphs that are similar
to each other, differing only by a few edges, in what we call 
\textit{generalizations}. 
This way, our proofs will avoid repeating the same 
arguments in similar 
situations. This idea reduces, with the help of a simplicity heuristic, the total 
computation time to less than 950 CPU days as shown
in Table~\ref{tab:reduction}.

From a set of graphs $\mathcal{G}$, we will construct a set of generalizations 
$\mathcal{G}^*$ (this is a set of set of graphs) with the following properties. 
Every graph in $\mathcal{G}$ is a member of a 
generalization in 
$\mathcal{G}^*$ (we are not missing any case) and every graph in a generalization 
$G^* 
\in \mathcal{G}^*$ is in $\mathcal{G}$ (we are not covering extra cases).

\begin{definition}(cover,exact cover)\\
A set of generalizations $\mathcal{G}^*$ is a cover of a set of graphs 
$\mathcal{G}$ if 
$\mathcal{G} \subseteq 
\bigcup_{G^*\in \mathcal{G}^*}  G^*$.\\
A set of generalizations $\mathcal{G}^*$ is an exact cover of
a set of graphs $\mathcal{G}$ if 
$\mathcal{G} = 
\bigcup_{G^*\in \mathcal{G}^*}G^*$.
\end{definition}

In a cover the generalizations do not need to be disjoint.
Furthermore, our proof does not fundamentally require the constructed 
covers to be exact and better covers may be obtained
by dropping this requirement. Yet, having exact covers
simplify our presentation of the gluing algorithm as it enable us
to ignore all clique constraints containing the splitting vertex 
(see Section~\ref{sec:gluing}).

\subsection{Algorithm for Constructing an Exact Cover}
In the following, we describe our base algorithm for constructing an
exact cover for a set of graphs $\mathcal{G}$. Our algorithm differs from
the one given in \cite{DBLP:journals/jgt/McKayR95} where they decide on which 
vertex to remove from a
graph. This is equivalent in our algorithm to ignoring the color of all edges 
connecting to that vertex. In 
contrast, our approach is more targeted and can decide whether to ignore the color
of an edge individually. Creating such alternative approach was crucial for us.
Indeed, following the original vertex removal method resulted in the creation of 
SAT problems, which were difficult to reconstruct in \holfour due to 
memory issues, negating most of the advantage gained by regrouping graphs.

 This will be
achieved by incrementally growing a set of generalizations 
$\mathcal{G}^*_\mathit{partial}$. We refer to the set of graphs $G\in\mathcal{G}$ 
that are not currently covered by $\mathcal{G}^*_\mathit{partial}$ as 
$\mathcal{G}_\mathit{uncovered}$.
Initially, $\mathcal{G}^*_\mathit{partial}$ is empty and thus 
$\mathcal{G}_\mathit{uncovered}$ is equal to $\mathcal{G}$.

At each iteration of our algorithm, we randomly pick a graph $G$ from 
$\mathcal{G}_\mathit{uncovered}$ constructing a singleton generalization 
$G^*_0 = \lbrace G \rbrace$.
Then, we color one of the edges of $G$ gray. This represents a generalization 
$G^*_1$ that contains the two graphs obtained by coloring the gray edge 
red or blue. Note that one of this graph is $G$ and thus $G \in G^*_1$ and $G^*_0
\subseteq G^*_1$. 
In general, the process starts from a generalization $G^*_n$ represented
as a graph with $n$ gray edges.
By definition, the generalization $G^*_n$ is defined to be the 
set of all graphs that can be obtained by coloring its $n$ gray edges red or blue 
in its representation.
Then, the algorithm selects randomly one edge to gray among edges respecting
the following conditions: the produced generalization $G_{n+1}^*$ must only contain
graphs that are in $\mathcal{G}$ and $(G_{n+1}^*\setminus G_{n})\ \cap\ 
\mathcal{G}_\mathit{uncovered}$ must contain at 
least $\lceil 2^{n-3} \rceil$ graphs. The first condition makes the cover
exact and the second condition prevents large overlaps between generalizations. 
The 
coefficient
$\lceil 2^{n-3} \rceil$ was experimentally determined and essentially
ensures that at least $\frac{1}{8}$ of the covered graphs by the newly created 
generalization are not covered by previous generalizations.

This process is repeated graying one more edge per generalization step, 
as illustrated in Figure~\ref{fig:gen}. It stops when the number of 
gray edges exceeds a user-given limit or when there are no more edges respecting 
the conditions. 

When the generalization algorithm stops, it creates a maximal generalization 
$G^*_{\mathit{max}}$ which is added to the
set of generalizations $\mathcal{G}^*_\mathit{partial}$ and the instantiations
of $G^*_{\mathit{max}}$ are removed from the set $\mathcal{G}_\mathit{uncovered}$
We keep adding new generalizations to $\mathcal{G}^*_\mathit{partial}$ by the
same procedure until the set $\mathcal{G}_\mathit{uncovered}$ is empty and 
therefore $\mathcal{G}^*_\mathit{partial}$ is an exact cover of $\mathcal{G}$.

We can reduce the size of the final cover $\mathcal{G}^*$ by sampling multiple 
graphs in $\mathcal{G}_\mathit{uncovered}$ at each iteration of the algorithm.
In our implementation, we sample 1000 graphs when $\mathcal{G}=\cR(4,4,k)$ 
and all graphs when $\mathcal{G}=\cR(3,5,k)$.
This produces one maximal generalization for each of those graphs.
We then select among them a generalization $G^*$ that contains 
a maximum number of uncovered graphs. That is to say one for which $
|G^* \cap \mathcal{G}_\mathit{uncovered}|$ is maximum. We call
this strategy for selecting generalization the \textit{greedy cover} strategy.
Our final strategy for constructing covers, described in Section~\ref{sec:better},
is a blend of the \textit{greedy cover} strategy and a strategy 
that minimizes the difficulty of resulting problems with respect to a simplicity 
heuristic given in Section~\ref{sec:simplicity}.

The \nauty algorithm~\cite{DBLP:journals/jsc/McKayP14} is called to normalize 
graphs in $\mathcal{G}$ and in each generalization $G^*$. By normalizing all 
graphs,
we can check that two graphs are isomorphic by simply checking if their 
normalizations are equal.

\begin{table}[t]
	\centering
	\begin{tabular}{rrrrr}\toprule
		k & $\cR(3,5,k)$ & $\cR^*(3,5,k)$ & 
		$\cR(4,4,k)$ & $\cR^*(4,4,k)$\\ 
		\midrule
		1 & 1 &  & 1 &  \\
		2 & 2 &  & 2 &  \\
		3 & 3 &  & 4 &  \\
		4 & 7 &  & 9 &  1\\
		5 & 13 & 3 & 24 &  3\\
		6 & 32 & 3 & 84 &  6\\
		7 & 71 & 5 & 362 &  11\\
		8 & \textbf{179} & \textbf{27} & 2079 &  47\\
		9 & 290 & 11 & 14701 & 271\\
		10 & \textbf{313} & \textbf{43} & 103706 & 1669 \\
		11 & 105 & 12 & 546356 & 7919\\
		12 & \textbf{12} & \textbf{12} & \textbf{1449166} & \textbf{26845} \\
		13 & 1 & 1 & 1184231 & 13078 \\
		14 &   &  & \textbf{130816} & \textbf{11752} \\
		15 &   &  & 640 & 67 \\
		16 &   &  & \textbf{2} & \textbf{2} \\
		17 &   &  & 1 & 1 \\
		\bottomrule
	\end{tabular}
	\caption{Number of $\cR(3,5,k)$-graphs and 
		$\cR(4,4,k)$-graphs up-to-isomorphism together with the number of 
		generalizations in the respective covers. All the covers were initially 
		constructed with a maximum of 10 gray edges.
		We later updated the cover for the bold cases using an
		edge selection algorithm and an improved selection algorithm for 
		generalizations (see Section~\ref{sec:better}). \label{tab:cover}
	    }
\end{table}
 
Computing the lists of $\cR(3,5,k)$-graphs and
$\cR(4,4,k)$-graphs up-to-isomorphism can be done efficiently by simply 
repeatedly extending 
graphs in $\cR(3,5,k)$ (resp.\ $\cR(4,4,k)$) by one vertex
while respecting the clique constraints to produce $\cR(3,5,k+1)$
(resp.\ $\cR(4,4,k+1)$). Such lists have been repeatedly compiled
as mentioned in~\cite{DBLP:journals/jgt/McKayR95} and therefore we will
not discuss in more detail how to construct them.
From those lists, we construct corresponding covers $\cR^*(3,5,k)$
and $\cR^*(4,4,k)$. The size of those constructed covers for the 
sets of graphs $\cR(3,5,k)$ with $5 \leq k \leq 13$ and 
the sets of graphs $\cR(4,4,k)$ with $4 \leq k \leq 17$ is presented
in Table~\ref{tab:cover}.

\begin{figure}
	
Iteration 0:\ $\mathcal{G}_\mathit{uncovered} = \mathcal{G}= \lbrace$
\begin{tikzpicture}[scale=0.35]
\draw[blue,thick] (0,0) -- (1,0);
\draw[blue,thick] (1,0) -- (0.5,0.866);
\draw[red,thick] (0.5,0.866) -- (0,0);
\end{tikzpicture}
,
\begin{tikzpicture}[scale=0.35]
	\draw[blue,thick] (0,0) -- (1,0);
	\draw[red,thick] (1,0) -- (0.5,0.866);
	\draw[red,thick] (0.5,0.866) -- (0,0);
\end{tikzpicture}
,
\begin{tikzpicture}[scale=0.35]
	\draw[red,thick] (0,0) -- (1,0);
	\draw[red,thick] (1,0) -- (0.5,0.866);
	\draw[red,thick] (0.5,0.866) -- (0,0);
\end{tikzpicture}
$\rbrace,$ \hspace{2mm} $\mathcal{G}^*_\mathit{partial} =\emptyset$\\

Randomly chosen generalization $G_0^*=\lbrace $ \begin{tikzpicture}[scale=0.35]
	\draw[blue,thick] (0,0) -- (1,0);
	\draw[red,thick] (1,0) -- (0.5,0.866);
	\draw[red,thick] (0.5,0.866) -- (0,0);
\end{tikzpicture} $\rbrace$

\vspace{3mm}

\begin{tikzpicture}
	\tikzset{
    node1/.style={
	shape=regular polygon,
	regular polygon sides=3,
	minimum size=1cm,
	inner sep=0,
	outer sep=0,
    path picture={
	\draw[red,line width=1mm] (path picture bounding box.south west) -- (path 
	picture 
	bounding box.north);
	\draw[red,line width=1mm] (path picture bounding box.north) -- (path picture 
	bounding 
	box.south east);
	\draw[blue,line width=1mm] (path picture bounding box.south east) -- (path 
	picture 
	bounding box.south west);
    },
			fill=none,
			draw=none
		}
    ,
        node4/.style={
    	shape=regular polygon,
    	regular polygon sides=3,
    	minimum size=1cm,
    	inner sep=0,
    	outer sep=0,
    	path picture={
    		\draw[blue,line width=1mm] (path picture bounding box.south west) -- 
    		(path 
    		picture 
    		bounding box.north);
    		\draw[red,line width=1mm] (path picture bounding box.north) -- (path 
    		picture 
    		bounding 
    		box.south east);
    		\draw[blue,line width=1mm] (path picture bounding box.south east) -- 
    		(path 
    		picture 
    		bounding box.south west);
    	},
    	fill=none,
    	draw=none
    }
    ,
        node5/.style={
    	shape=regular polygon,
    	regular polygon sides=3,
    	minimum size=1cm,
    	inner sep=0,
    	outer sep=0,
    	path picture={
    		\draw[red,line width=1mm] (path picture bounding box.south west) -- 
    		(path 
    		picture 
    		bounding box.north);
    		\draw[red,line width=1mm] (path picture bounding box.north) -- (path 
    		picture 
    		bounding 
    		box.south east);
    		\draw[red,line width=1mm] (path picture bounding box.south east) -- 
    		(path 
    		picture 
    		bounding box.south west);
    	},
    	fill=none,
    	draw=none
    }
    ,
	node2/.style={
		shape=regular polygon,
		regular polygon sides=3,
		minimum size=1cm,
		inner sep=0,
		outer sep=0,
		path picture={
			\draw[gray,line width=1mm] (path picture bounding box.south 
			west) -- 
			(path 
			picture 
			bounding box.north);
			\draw[red,line width=1mm] (path picture bounding box.north) -- (path 
			picture 
			bounding 
			box.south east);
			\draw[blue,line width=1mm] (path picture bounding box.south east) -- 
			(path 
			picture 
			bounding box.south west);
		},
		fill=none,
		draw=none
	}
    ,
node3/.style={
	shape=regular polygon,
	regular polygon sides=3,
	minimum size=1cm,
	inner sep=0,
	outer sep=0,
	path picture={
		\draw[gray,line width=1mm] (path picture bounding box.south 
		west) -- 
		(path 
		picture 
		bounding box.north);
		\draw[red,line width=1mm] (path picture bounding box.north) -- 
		(path 
		picture 
		bounding 
		box.south east);
		\draw[gray,line width=1mm] 
		(path picture bounding box.south east) -- 
		(path 
		picture 
		bounding box.south west);
	},
	fill=none,
	draw=none
	}
    }	
	
	\node[node1] (l0) {$G_0^*$};
	\node[left of=l0, node distance=2.5cm] (l0l) {\textbf{Generalizations}}
	node[node2, node distance =3cm, right of=l0] (l1){$G_1^*$}; 
	\node[node3, node distance =3cm, right of=l1] (l2){$G_2^*$}; 
	\node[node distance = 2cm, right of=l2] (l3){}; 
	\node[node1, node distance = 2cm, below of=l0] (l0b) {};
	\node[left of=l0b, node distance=2.5cm] (l0bl) {\textbf{Graphs}};
	\node[node4, node distance = 2cm, below of=l1] (l1b) {};
	\node[node5, node distance = 2cm, below of=l2] (l2b) {};
	
	\draw[->,thick] (l0) -- (l1) node[midway, above] {};
	\draw[->,thick] (l1) -- (l2) node[midway, above] {};
	\draw[->,thick] (l2) -- (l3) node[midway, above] {stops};
	
	\draw[<-,dotted,thick] (l0) -- (l0b) {};
	\draw[<-,dotted,thick] (l1) -- (l0b) {};
	\draw[<-,dotted,thick] (l1) -- (l1b) {};
	\draw[<-,dotted,thick] (l2) -- (l0b) {};
	\draw[<-,dotted,thick] (l2) -- (l1b) {};
	\draw[<-,dotted,thick] (l2) -- (l2b) {};
	
\end{tikzpicture}

\vspace{3mm}

Iteration 1: $\mathcal{G}_\mathit{uncovered}=\emptyset$, \hspace{2mm} 
$\mathcal{G}^* = \mathcal{G}^*_\mathit{partial}=\lbrace G_2^*\rbrace=\lbrace$ 
\begin{tikzpicture}[scale=0.35]
	\draw[gray,thick] (0,0) -- (1,0);
	\draw[red,thick] (1,0) -- (0.5,0.866);
	\draw[gray,thick] (0.5,0.866) -- (0,0);
\end{tikzpicture}
$\rbrace$
\caption {Construction of an exact cover $\mathcal{G^*}$ of a set of graphs
$\mathcal{G}$. 
The process of graying edges stops as it 
would otherwise produce a gray triangle including a blue triangle. 
The construction of an exact cover terminates in this case after one iteration.
The dotted arrows indicate which graph belongs to which 
generalization.\label{fig:gen}}
\end{figure}

\subsection{Proof that 
\texorpdfstring{$\cR(3,5,k)$}{R(3,5,k)} is Covered by 
\texorpdfstring{$\cR^*(3,5,k)$}{R*(3,5,k)}}\label{sec:proofenum}

In this section, we only present the proof
for the covers $\cR^*(3,5,k)$ since the proof for the covers $\cR^*(4,4,k)$  
follows by an analogous argument.
Given the result presented in Section~\ref{sec:deg},
it is enough to consider the cases of a splitting vertex with degree 
$d \in \{8, 10, 12\}$. Therefore, it would 
be enough to prove that $\cR^*(3,5,d)$ covers 
the set of graphs with property $\cR(3,5,d)$.
However, to do so, we found it easier to prove the stronger result:
\[\forall\ 5 \leq k \leq 13.\ 
G \mbox{\ has\ property\ } \cR(3,5,k) 
\Rightarrow \exists G^* \in \cR^*(3,5,k).\ G \in G^*
\]

We prove the result by a finite induction over the number of vertices $k$.\\
The base case $k=5$ consists of searching for all the possible graphs 
with property $\cR(3,5,5)$ and show that they appear modulo isomorphism in 
one of the generalizations in $\cR^*(3,5,5)$.\\
The inductive case is similar. The main difference is that we start the search 
from a generalization $G^*$ instead of the empty generalization.
We prove that for all generalizations $G^*$ in $\cR^*(3,5,k)$, any 
extension of $G^*$ by one vertex that respects the property $\cR^*(3,5,k+1)$ is 
isomorphic to an element of $\cR^*(3,5,k+1)$. This is achieved by exploring
all possible colorings (in blue or in red) of edges that are either gray or contain
the new vertex. This extension process is depicted in Figure~\ref{fig:ext}.
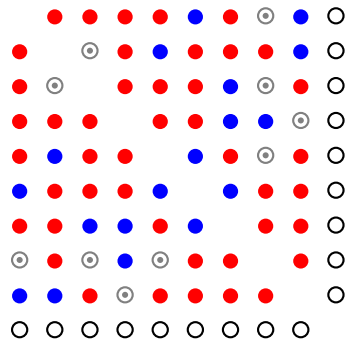
\begin{figure}
	\centering
	\begin{tikzpicture}[]
		\matrix (adjacency) [matrix of math nodes] at (0,0)
		{ & \redr & \redr & \redr & \redr & \blub & \redr & \grag & \blub & \trat 
		\\
		\redr &   & \grag & \redr & \blub & \redr & \redr & \redr & \blub & \trat  
		\\
		\redr & \grag &   & \redr & \redr & \redr & \blub & \grag & \redr & \trat  
		\\
		\redr & \redr & \redr &   & \redr & \redr & \blub & \blub & \grag & \trat  
		\\
		\redr & \blub & \redr & \redr &   & \blub & \redr & \grag & \redr & \trat  
		\\
		\blub & \redr & \redr & \redr & \blub &   & \blub & \redr & \redr & \trat  
		\\
		\redr & \redr & \blub & \blub & \redr & \blub &   & \redr & \redr  & \trat 
		\\
		\grag & \redr & \grag & \blub & \grag & \redr & \redr &   & \redr  & \trat 
		\\
		\blub & \blub & \redr & \grag & \redr & \redr & \redr & \redr &     & 
		\trat \\
		\trat & \trat & \trat & \trat & \trat & \trat & \trat & \trat & \trat &
		\\};
	\end{tikzpicture}
	\caption{Extension of an $\cR^*(3,5,9)$-generalization depicted as
		an adjacency matrix.
		The first 9 rows and columns represent
		the vertices $x_0$ to $x_8$ of the $\cR^*(3,5,9)$-generalization.
		Gray edges are represented by dotted gray circles.
		Edges containing the extension vertex $x_9$ (last row and column) are 
		represented by 
		black circles.
		\label{fig:ext}}
\end{figure}

The formalization and efficiency of the previous arguments rely
on our custom-made solver for labeled graphs. Our solver mostly works like
a DPLL SAT solver~\cite{DBLP:journals/jacm/DavisP60}.
The principal difference is that it represents clauses as essentially 
first-order formulas.
We show how we represent the property  $\cR^*(3,5,k+1)$ in first-order.
Given a graph $G$ of size $k+1$, represented by
a binary relation $E$ over a set of vertices $V=[|0,k|]=\{ 0,1,\ldots,k 
\}$, our 
first-order 
representation of the statement ``$G$ has property $\cR^*(3,5,k+1)$'' is
given by the two formulas:
\begin{align*}
	\forall\underset{\mathit{distinct}}{x_0 x_1 x_2} < k \!+\! 1.\ &
	\neg E x_0 x_1 \vee \neg E x_0 x_2 \vee \neg E x_1 x_2\\
	\forall\underset{\mathit{distinct}}{x_0 x_1 x_2 x_3 x_4} < k \!+\! 1.\ &
	 E x_0 x_1 \vee E x_0 x_2 \vee E x_0 x_3 \vee E x_0 x_4 \vee E x_1 x_2 \vee \\
	&E x_1 x_3 \vee E x_1 x_4 \vee E x_2 x_3 \vee E x_2 x_4 \vee E x_3 x_4
\end{align*}

where $\mathit{distinct}$ means that we add inequalities between each of pair of 
quantified variables. For example,
$\forall\underset{\mathit{distinct}}{x_0 x_1 x_2} < k+1.\ P[x_1,x_2,x_3]$ 
stands
for: \[\forall x_0 x_1 x_2.\ 
(x_0 \!<\! k \!+\! 1 \wedge x_1 \!<\! k \!+\! 1 \wedge x_2 \!<\! k \!+\! 1 \wedge
x_0 \!\neq\! x_1 \wedge x_0 \!\neq\! x_2 \wedge x_1 \!\neq\! x_2) \Rightarrow 
P[x_1,x_2,x_3]\]

Our solver is designed to work exclusively on graphs 
where edges are represented by first-order literals and the coloring of an edge 
$(i,j)$ to blue (resp.\ red) corresponds to assuming the literal $E x_i x_j$ in 
the 
branch (resp.\ the negation of the literal $\neg E x_i x_j$).
A gray edge just indicates that no color for that edge is currently assumed.
The prover starts by assuming all literals corresponding
to colored edges in a generalization $G^*$ with $k$ vertices. 
It then explores all possible colorings of the gray edges and then all the 
possible coloring of the edges containing a new vertex $x_k$.
At the leaf, this produces a graph of size $k+1$ represented by all
the literals assumed in the branch.

By representing vertices by variables $x_0,\ldots,x_k$ instead
of the concrete value $0,\ldots,k$, it is easier to prove that the graphs in the 
leaves are indeed isomorphic to one of the elements of one of the generalizations
$G' \in \cR^*(3,5,k+1)$. 
To prove that every generalization in $\cR^*(3,5,k)$
extends to graphs belonging to a generalization in 
$\cR^*(3,5,k+1)$, we will prove that there is no way to extend a generalization in 
$\cR^*(3,5,k)$ if we forbid the creation of any graph that is a 
member of a generalization in $\cR^*(3,5,k+1)$.
Given a generalization $G' \in \cR^*(3,5,k+1)$ with a set of 
blue edges $\mathit{Blue}$ and a set of 
red edges $\mathit{Red}$, we use the following formula to forbid the creation of 
an 
element 
of $G'$:
\[\forall \underset{\mathit{distinct}}{x_0 \ldots x_i \ldots x_j \ldots x_k} < 
k+1.\ \ 
((\bigwedge\limits_{(i,j) \in \mathit{Blue}}
E x_i x_j)\ 
\wedge \
(\bigwedge\limits_{(i,j) \in \mathit{Red}}
\neg E x_i x_j)) \Rightarrow \bot
\]
Note that this implies that all permutations of graphs that are members of
this generalization are forbidden. 
The first reason is that the formula does not assume any constraints on the gray 
edges of $G'$ therefore forbids all members of $G'$. 
The second reason is one can permute the 
indices of variables by a simple 
instantiation of the variables with a permutation being given by the 
\nauty algorithm to make the labeled graph on the branch match
with one of the labeled generalizations.

\section{Gluing}~\label{sec:gluing}
In the previous section, we constructed covers for $\cR(3,5,d)$-graphs 
and $\cR(4,4,24-d)$-graphs. 
The next step of our proof is to prove
that given a generalization $G^*$ in $\cR^*(3,5,d)$ and a generalization in $H^*$ 
in $\cR^*(3,5,24-d)$, there is no way to extend color gray edges and transverse 
edges to form an $\cR(4,5,24)$-graph (see Figure~\ref{fig:hi}) and thus an 
$\cR(4,5,25)$-graph by 
adding the splitting vertex.
In the rest of this section, we can ignore
clique constraints that include the splitting vertex as they are already
satisfied. This is a consequence of the fact that our covers are exact covers.
All our gluing problems are formulated at the propositional level and 
contain the following clauses representing the property $\cR(4,5,24)$.
Let us number the vertices of an $\cR^*(3,5,d)$-generalization $G^*$ from $0$ to 
$d-1$, and the vertices of a  $\cR^*(4,4,24-d)$-generalization graph $H^*$ from 
$d$ to $23$.\\
\[
\mbox{For each subset $S \subset [|0,23|]$ of size 4, we create the clause }
\bigvee\limits_{a,b \in S \wedge a < b}\neg E_{a,b}\ .\]
\[
\mbox{For each subset $T \subset [|0,23|]$ of size 5, we create the clause 
}\bigvee\limits_{a,b \in T \wedge a < b} E_{a,b}\ .\]

In all these propositional clauses, $E_{a,b}$ 
is a propositional variable that is true if there is a blue edge between $a$ and
$b$ and
that is false if there is a red edge between $a$ and
$b$.
One can note that any clauses containing only vertices from $G^*$ or only vertices 
from $H^*$ can be omitted as $G^*$ and $H^*$ provably avoid any blue 
4-clique or any red 5-clique. This removal occurs naturally as a consequence of 
performing unit propagation.
In each gluing problem ($G^*$,$H^*$), we add
unit clauses for each colored edge (red or blue) of $G^*$ and $H^*$. 
If an edge $(a,b)$ with $a<b$ is blue 
then we add the unit clause $E_{a,b}$, if it is
red then we add the unit clause $\neg E_{a,b}$. If an edge is gray we
do not add a unit clause. Together, with the clique clauses this forms our SAT 
problem that is sent to the \minisat interface. 
In practice, we had to
perform unit propagation to reduce the number of clauses before sending a 
problem to the interface. This is due to some limitations in the 
interface as this does not happen when we call the SAT solver directly.

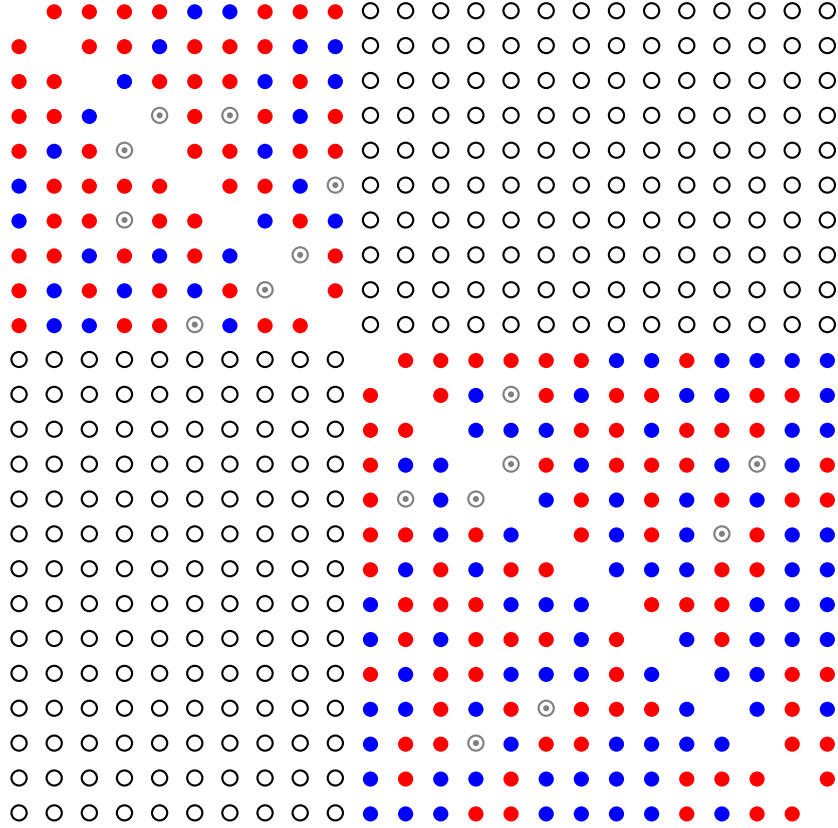
\begin{figure}[ht]
	\centering
 \begin{tikzpicture}[]
	\matrix (adjacency) [matrix of math nodes] at (0,0)
	{
   & \redr & \redr & \redr & \redr & \blub & \blub & \redr & \redr & \redr & 
\trat & \trat & \trat & \trat & \trat & \trat & \trat & \trat & \trat & \trat & 
\trat & \trat & \trat & \trat \\
\redr &   & \redr & \redr & \blub & \redr & \redr & \redr & \blub & \blub & 
\trat & \trat & \trat & \trat & \trat & \trat & \trat & \trat & \trat & \trat & 
\trat & \trat & \trat & \trat \\
 \redr & \redr &   & \blub & \redr & \redr & \redr & \blub & \redr & \blub & 
\trat & \trat & \trat & \trat & \trat & \trat & \trat & \trat & \trat & \trat & 
\trat & \trat & \trat & \trat \\
 \redr & \redr & \blub &   & \grag & \redr & \grag & \redr & \blub & \redr & 
\trat & \trat & \trat & \trat & \trat & \trat & \trat & \trat & \trat & \trat & 
\trat & \trat & \trat & \trat \\
\redr & \blub & \redr & \grag &   & \redr & \redr & \blub & \redr & \redr & 
\trat & \trat & \trat & \trat & \trat & \trat & \trat & \trat & \trat & \trat & 
\trat & \trat & \trat & \trat \\
 \blub & \redr & \redr & \redr & \redr &   & \redr & \redr & \blub & \grag & 
\trat & \trat & \trat & \trat & \trat & \trat & \trat & \trat & \trat & \trat & 
\trat & \trat & \trat & \trat \\
 \blub & \redr & \redr & \grag & \redr & \redr &   & \blub & \redr & \blub & 
\trat & \trat & \trat & \trat & \trat & \trat & \trat & \trat & \trat & \trat & 
\trat & \trat & \trat & \trat \\
 \redr & \redr & \blub & \redr & \blub & \redr & \blub &   & \grag & \redr & 
\trat & \trat & \trat & \trat & \trat & \trat & \trat & \trat & \trat & \trat & 
\trat & \trat & \trat & \trat \\
\redr & \blub & \redr & \blub & \redr & \blub & \redr & \grag &   & \redr & 
\trat & \trat & \trat & \trat & \trat & \trat & \trat & \trat & \trat & \trat & 
\trat & \trat & \trat & \trat \\
 \redr & \blub & \blub & \redr & \redr & \grag & \blub & \redr & \redr &   & 
\trat & \trat & \trat & \trat & \trat & \trat & \trat & \trat & \trat & \trat & 
\trat & \trat & \trat & \trat \\
 \trat & \trat & \trat & \trat & \trat & \trat & \trat & \trat & \trat & \trat 
&   & \redr & \redr & \redr & \redr & \redr & \redr & \blub & \blub & \redr & 
\blub & \blub & \blub & \blub \\
 \trat & \trat & \trat & \trat & \trat & \trat & \trat & \trat & \trat & \trat 
& \redr &   & \redr & \blub & \grag & \redr & \blub & \redr & \redr & \blub & 
\blub & \redr & \redr & \blub \\
 \trat & \trat & \trat & \trat & \trat & \trat & \trat & \trat & \trat & \trat 
& \redr & \redr &   & \blub & \blub & \blub & \redr & \redr & \blub & \redr & 
\redr & \redr & \blub & \blub \\
 \trat & \trat & \trat & \trat & \trat & \trat & \trat & \trat & \trat & \trat 
& \redr & \blub & \blub &   & \grag & \redr & \blub & \redr & \redr & \redr & 
\blub & \grag & \blub & \redr \\
 \trat & \trat & \trat & \trat & \trat & \trat & \trat & \trat & \trat & \trat 
& \redr & \grag & \blub & \grag &   & \blub & \redr & \blub & \redr & \blub & 
\redr & \blub & \redr & \redr \\
 \trat & \trat & \trat & \trat & \trat & \trat & \trat & \trat & \trat & \trat 
& \redr & \redr & \blub & \redr & \blub &   & \redr & \blub & \redr & \blub & 
\grag & \redr & \blub & \blub \\
 \trat & \trat & \trat & \trat & \trat & \trat & \trat & \trat & \trat & \trat 
& \redr & \blub & \redr & \blub & \redr & \redr &   & \blub & \blub & \blub & 
\redr & \redr & \blub & \blub \\
 \trat & \trat & \trat & \trat & \trat & \trat & \trat & \trat & \trat & \trat 
& \blub & \redr & \redr & \redr & \blub & \blub & \blub &   & \redr & \redr & 
\redr & \blub & \blub & \blub \\
 \trat & \trat & \trat & \trat & \trat & \trat & \trat & \trat & \trat & \trat 
& \blub & \redr & \blub & \redr & \redr & \redr & \blub & \redr &   & \blub & 
\redr & \blub & \blub & \blub \\
 \trat & \trat & \trat & \trat & \trat & \trat & \trat & \trat & \trat & \trat 
& \redr & \blub & \redr & \redr & \blub & \blub & \blub & \redr & \blub &   & 
\blub & \blub & \redr & \redr \\
 \trat & \trat & \trat & \trat & \trat & \trat & \trat & \trat & \trat & \trat 
& \blub & \blub & \redr & \blub & \redr & \grag & \redr & \redr & \redr & \blub 
&   & \blub & \redr & \blub \\
 \trat & \trat & \trat & \trat & \trat & \trat & \trat & \trat & \trat & \trat 
& \blub & \redr & \redr & \grag & \blub & \redr & \redr & \blub & \blub & \blub & 
\blub &   & \redr & \redr \\
 \trat & \trat & \trat & \trat & \trat & \trat & \trat & \trat & \trat & \trat 
& \blub & \redr & \blub & \blub & \redr & \blub & \blub & \blub & \blub & \redr & 
\redr & \redr &   & \redr \\
 \trat & \trat & \trat & \trat & \trat & \trat & \trat & \trat & \trat & \trat 
& \blub & \blub & \blub & \redr & \redr & \blub & \blub & \blub & \blub & \redr & 
\blub & \redr & \redr &   \\};
\end{tikzpicture}
\caption{The adjacency matrix of a graph of size 24
where a partial coloring is given by a
 generalization $G^*$ with 4 gray edges (dotted gray circles) with vertices 
 numbered from 
 0 to 9 and a 
 generalization $H^*$ with 4 gray edges with vertices 
 numbered from 
 10 to 23.
 The goal of the SAT solvers is to prove that there is no way to assign a
 color (blue or red) to the gray edges and the transverse edges (black circles) 
 without creating a blue 4-clique or a red 5-clique. 
 \label{fig:hi}
}
\end{figure}

\subsection{Simplicity Heuristic}~\label{sec:simplicity}
In this section, we design a heuristic that will be used
to construct covers resulting in easier problems for the SAT solver.
Let $b_{k}$ represent the number of blue 
$k$-cliques in $G^*$.
Let $r_{k}$ represent the number of red $k$-cliques in $G^*$.
Let $b'_{k}$ represent the number of blue $k$-cliques in $H^*$.
Let $r'_{k}$ represent the number of red $k$-cliques in $H^*$.
We use the following formula to estimate the difficulty of a gluing problem 
$(G^*,H^*)$:
\[\mathit{simplicity}(G^*,H^*) =\frac{1}{2^3} b_1 b'_3 + \frac{1}{2^4}b_2 b'_2 +
\frac{1}{2^6} r_2 r'_3 + \frac{1}{2^6}r_3 r'_2 + 
\frac{1}{2^4} r_4 r'_1\]
Our formula is originally designed to estimate the simplicity of
a problem of gluing an $\cR(3,5,d)$-graph $G$ with an $\cR(4,4,24-d)$-graph $H$. 
There, 5 different types of configurations that may create a blue 4-clique or a 5 
red-clique as illustrated in Figure~\ref{fig:config}. In the resulting SAT solving
problem after unit propagation, a clause mentioning only the transverse edges is 
associated with each configuration. 
The above heuristic can be derived from a more general heuristic for a SAT problem 
$P$: 
\[\mathit{simplicity}(P) = \sum_{c \in P} \frac{1}{2^{|c|}}\]
where $|c|$ is the number of literals in a clause $c$.
This heuristic operates under the simplistic assumption that each clause covers 
separate cases, allowing it to estimate the extent of search space coverage by 
summing up the contribution of each clause.
Experimentally, we found that this heuristic is only effective to compare problems 
with the same number of variables. Consequently, we ignore clauses containing gray 
edges from all gluing problems when computing their simplicity, as we will use
this heuristic to compare generalizations with varying numbers of gray edges.
Another advantage of ignoring clauses containing gray edges is that the heuristic 
will prefer problems that can delay
splitting on the color of gray edges as much as possible, which favors
proof sharing.

\begin{figure}[ht]
\begin{center}
\begin{tikzpicture}[]
	\coordinate (bG11) at (0,0.5);
	\coordinate (bH31) at (3,0);
	\coordinate (bH32) at (2.134,0.5);
	\coordinate (bH33) at (3,1);
	\draw[blue, thick] (bH31) -- (bH32);
	\draw[blue, thick] (bH32) -- (bH33);
	\draw[blue, thick] (bH33) -- (bH31);
	\draw[dotted,thick] (bG11) -- (bH31);
	\draw[dotted,thick] (bG11) -- (bH32);
	\draw[dotted,thick] (bG11) -- (bH33);
	
	\foreach \point in {bG11,bH31,bH32,bH33}
	\node[draw, circle, fill=black, inner sep=1pt] at (\point) {};
\end{tikzpicture}	
\hspace{10mm}	
\begin{tikzpicture}[]
	\coordinate (bG21) at (0,0);
	\coordinate (bG22) at (0,1);
	\coordinate (bH21) at (3,0);
	\coordinate (bH22) at (3,1);
	\draw[blue, thick] (bG21) -- (bG22);
	\draw[blue, thick] (bH21) -- (bH22);
	\draw[dotted,thick] (bG21) -- (bH21);
	\draw[dotted,thick] (bG22) -- (bH22);
	\draw[dotted,thick] (bG21) -- (bH22);
	\draw[dotted,thick] (bG22) -- (bH21);
	
	\foreach \point in {bG21,bG22,bH21,bH22}
	\node[draw, circle, fill=black, inner sep=1pt] at (\point) {};
\end{tikzpicture}		
\end{center}

\begin{center}	
\begin{tikzpicture}[]
\coordinate (rG31) at (3,0);
\coordinate (rG32) at (3,1);
\coordinate (rG33) at (2.134,0.5);
\coordinate (rH21) at (0,0);
\coordinate (rH22) at (0,1);
\draw[red, thick] (rG31) -- (rG32);
\draw[red, thick] (rG32) -- (rG33);
\draw[red, thick] (rG33) -- (rG31);
\draw[red, thick] (rH21) -- (rH22);
\draw[dotted,thick] (rG31) -- (rH21);
\draw[dotted,thick] (rG31) -- (rH22);
\draw[dotted,thick] (rG32) -- (rH21);
\draw[dotted,thick] (rG32) -- (rH22);
\draw[dotted,thick] (rG33) -- (rH21);
\draw[dotted,thick] (rG33) -- (rH22);

\foreach \point in {rG31,rG32,rG33,rH21,rH22}
\node[draw, circle, fill=black, inner sep=1pt] at (\point) {};
\end{tikzpicture}	
\hspace{10mm}
\begin{tikzpicture}[]
	
\coordinate (rG31) at (0,0);
\coordinate (rG32) at (0,1);
\coordinate (rG33) at (0.866,0.5);
\coordinate (rH21) at (3,0);
\coordinate (rH22) at (3,1);

\draw[red, thick] (rG31) -- (rG32);
\draw[red, thick] (rG32) -- (rG33);
\draw[red, thick] (rG33) -- (rG31);
\draw[red, thick] (rH21) -- (rH22);
\draw[dotted,thick] (rG31) -- (rH21);
\draw[dotted,thick] (rG31) -- (rH22);
\draw[dotted,thick] (rG32) -- (rH21);
\draw[dotted,thick] (rG32) -- (rH22);
\draw[dotted,thick] (rG33) -- (rH21);
\draw[dotted,thick] (rG33) -- (rH22);

\foreach \point in {rG31,rG32,rG33,rH21,rH22}
\node[draw, circle, fill=black, inner sep=1pt] at (\point) {};
\end{tikzpicture}
\hspace{10mm}
\begin{tikzpicture}[]
	
	\coordinate (rG41) at (0,0);
	\coordinate (rG42) at (0,1);
	\coordinate (rG43) at (1,1);
	\coordinate (rG44) at (1,0);
	\coordinate (rH11) at (3,0.5);
	
	\draw[red, thick] (rG41) -- (rG42);
	\draw[red, thick] (rG42) -- (rG43);
	\draw[red, thick] (rG43) -- (rG44);
	\draw[red, thick] (rG44) -- (rG41);
	\draw[red, thick] (rG41) -- (rG43);
	\draw[red, thick] (rG42) -- (rG44);
	\draw[dotted,thick] (rG41) -- (rH11);
	\draw[dotted,thick] (rG42) -- (rH11);
	\draw[dotted,thick] (rG43) -- (rH11);
	\draw[dotted,thick] (rG44) -- (rH11);

	\foreach \point in {rG41,rG42,rG43,rG44,rH11}
	\node[draw, circle, fill=black, inner sep=1pt] at (\point) {};
\end{tikzpicture}
\end{center}
\caption{The five possible types of configurations.
	    In each configuration, a colored clique in an $\cR(3,5,d)$-graph is 
	    displayed on the left and 
	    a colored clique in an $\cR(4,4,24-d)$-graph is displayed on the right.
	    Transverse edges are shown as dotted black edges.
	    Transverse edges must not all be blue in blue configurations
	     and they must not all be red in red configurations.
	     \label{fig:config}}
\end{figure}
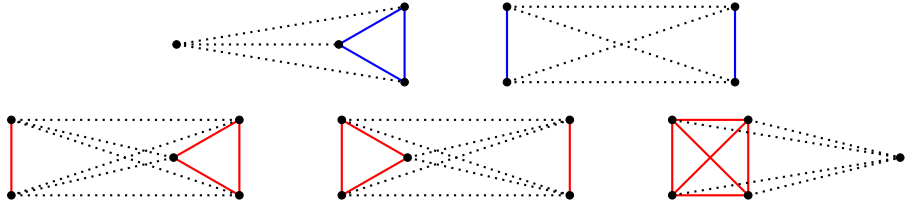
%

We now test how good the simplicity score is at predicting the run
time of \minisat via the \holfour interface on 
200 gluing problems between
$\cR(3,5,10)$-graphs and $\cR(4,4,14)$-graphs in Figure~\ref{fig:simp}. The 
results reveal that our simplicity score is a good predictor in this setting.

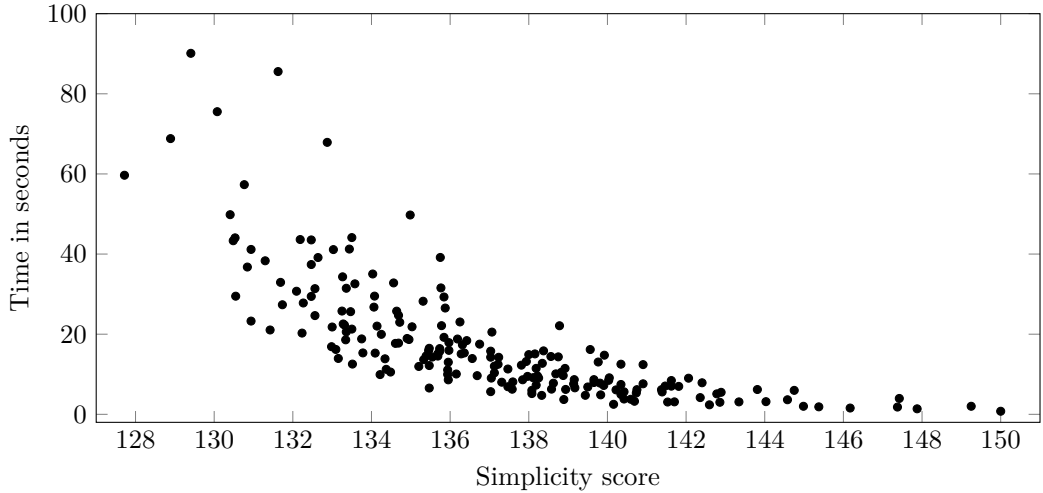
\begin{figure}[ht]
	\centering
	\begin{tikzpicture}[scale=1.0]
			\begin{axis}[
					legend style={anchor=north east, at={(0.95,0.4)}},
					width=\textwidth,
					height=0.5*\textwidth,xmin=127, xmax=151,
					ymin=-2, ymax=100,
					xlabel = {Simplicity score},
					ylabel = {Time in seconds} 
					]
					\addplot[mark size=1.5pt,only marks] table[x=information, 
		y=time] 
					{heuristic};
				\end{axis}
		\end{tikzpicture}
	\caption{Relation between the simplicity score of a problem and the 
			time required by the \holfour interface to \minisat to prove that it is
			unsatisfiable.
			Each problem consists of an $\cR(3,5,10)$-graph and 
			an $\cR(4,4,14)$-graph.
			Each dot represents one problem among a random sample of 200 gluing 
			problems. \label{fig:simp}}
\end{figure}

Finally, during the construction of a cover for $\mathcal{G}$ we are not aware of 
the 
corresponding cover for $\mathcal{H}$. The covers would have to be built 
simultaneously making the algorithm more complicated.
To avoid those complications, 
we devise a measure to predict if a generalization will create 
difficult problems on its own without depending on the possible counterparts.
To this end, we chose to estimate the simplicity of a generalization $G^*$ by 
how difficult it is to glue it with an average counterpart graph $\bar{H}$. 
Let $\bar{b}'_{k}$ represent the average number of blue $k$-cliques 
per graph in $\mathcal{H}$
and $\bar{r}'_{k}$ represent the average number of red $k$-cliques per graph in 
$\mathcal{H}$,
the simplicity of $G^*$ is:
\[\mathit{simplicity}(G^*) = \frac{1}{2^3} b_1 \bar{b}'_3 + \frac{1}{2^4}b_2 
\bar{b}'_2 +
\frac{1}{2^6} r_2 \bar{r}'_3 + \frac{1}{2^6}r_3 \bar{r}'_2 + 
\frac{1}{2^4} r_4 \bar{r}'_1\]
Similarly, 
let $\bar{b}_{k}$ represent the average number of blue $k$-cliques per graph in 
$\mathcal{G}$
and $\bar{r}_{k}$ represent the average number of red $k$-cliques per graph in 
$\mathcal{G}$,
the simplicity of $H^*$ is:
\[\mathit{simplicity}(H^*) = \frac{1}{2^3} \bar{b}_1 b'_3 + \frac{1}{2^4}\bar{b}_2 
b'_2 +
\frac{1}{2^6} \bar{r}_2 r'_3 + \frac{1}{2^6}\bar{r}_3 r'_2 + 
\frac{1}{2^4} \bar{r}_4 r'_1\]

\subsection{Creation of Better Covers by Parameter Search}\label{sec:better}

We now improve the exact cover algorithm by relying on our simplicity heuristic
in two places. The first one is during the selection of edges. Previously,
the edges were selected at random as long as the produced generalization
satisfied some conditions. Now, we will select, among the possible edges allowed
by the conditions, the one
that produces the generalization with the highest simplicity score. We call
this strategy \textit{fastest} in Table~\ref{tab:test}. The second place where the 
simplicity 
score will influence the algorithm is during the selection of generalizations. 
Previously, we selected maximal generalizations $G^*$ with highest coverage
value $n_\mathit{cover}(G^*) = | G^* \cap \mathcal{G_\mathit{uncover}} |$. 
This strategy was called the \textit{greedy cover} strategy.
Now, we will also prefer generalizations with higher simplicity scores. Since we 
want to optimize for both objectives at the same time, we will select the graph 
$G^*$ with the highest combined score $\mathit{simplicity}(G^*)^{c} \times 
n_\mathit{cover}(G^*)$ where $\mathit{c}$ is a real number parameter 
influencing how much one heuristic is preferred over the other. 
In Table~\ref{tab:test}, we call this selection strategy \textit{mixed-c}.

Although the simplicity score is important to reduce the difficulty of the 
problems,
the most important parameter in reducing the total computation time is the number
of maximum allowed gray edges in each generalization.
In Table~\ref{tab:test}, we optimize those parameters for the three relevant 
cases $d=8,10,12$ for the gluing. Each experiment consists of a line in 
Table~\ref{tab:test}. 
There, we compute new covers with different parameters. To figure out the best 
parameters, we estimate the run time of an 
average problem by sampling 200 random problems from a pair of covers.
We then multiply this estimate by the number of problems this pair of covers would 
create to get an estimated total run time for this pair of covers.
In the end, we decided not to go with the best parameters according to the
estimated times given in Table~\ref{tab:test}.
The reason is that by increasing the number of gray edges, the total number of 
problems is reduced but the difficulty of each problem is increased.
This makes the problems harder and they would have taken more memory than
we had available. That is why we chose to compromise and instead use the fastest 
parameter settings, shown in bold in Table~\ref{tab:test}, 
that would not use more memory than available on our machines.
Table~\ref{tab:reduction} gives a comparison of the total run time for our final 
problems with the total runtime that we would get by simply gluing pairs of graphs 
instead 
of generalizations. For degree $d = 10$, the number of gluing problems
is reduced by a factor of $81.0$ and the estimated total time by a factor $14.6$.
For degree $d = 12$, the number of gluing problems is reduced by a factor of 
$53.9$ and the estimated total time by a factor of $20.6$.

\begin{table}[ht]
	\centering
	\begin{tabular}{llllll}
		\toprule
		Gluing & 3,5 & 4,4 & Edge sel. & Gen. sel. & CPU-days (estimation)\\
		\midrule
		3,5,8-4,4,16 & 0 & 0 & none & none &  0.055 \\
		             & \textbf{4} & \textbf{0} & \textbf{fastest} & 
		             \textbf{mixed-0.5} & \textbf{0.018} \\
		\midrule
		3,5,10-4,4,14 & 0 & 0 & none & none & 8373 \\
		& 3 & 4 & fastest & mixed-1.0 & 725\\
		& 4 & 3 & fastest & mixed-1.0 & 689\\
		& 4 & 4 & random & greedy cover & 734\\
		& 4 & 4 & fastest & mixed-10.0 & 625\\
		& 4 & 4 & fastest & mixed-2.0 & 595\\
		& 4 & 4 & fastest & mixed-1.0 & 658\\
		& \textbf{4} & \textbf{4} & \textbf{fastest} & \textbf{mixed-0.5} 
				& 
		\textbf{572}\\
		& 4 & 4 & fastest & mixed-0.1 & 706\\	
		& 5 & 4 & fastest & mixed-1.0 & 547\\
		& 4 & 5 & fastest & mixed-1.0 & 586\\
		& 5 & 5 & fastest & mixed-0.5 & 396\\
		\midrule
		3,5,12-4,4,12 & 0 & 0 & none & none & 7702 \\
		& 2 & 6 & fastest & mixed-0.5 & 641\\
		& 3 & 6 & fastest & mixed-0.5 & 782\\
		& 4 & 6 & fastest & mixed-0.5 & 784\\		
		& \textbf{0} & \textbf{8} & \textbf{fastest} & \textbf{mixed-0.5} & 
		\textbf{374}\\
		& 1 & 8 & fastest & mixed-0.5 & 353\\
		& 2 & 8 & fastest & mixed-0.5 & 538\\
		& 3 & 8 & fastest & mixed-0.5 & 360\\
		& 4 & 8 & fastest & mixed-0.5 & 419\\
		\bottomrule 
	\end{tabular}
	\caption{Tested parameters for creating exact covers. The columns titled 3,5 
	and 4,4 show the maximum number of gray edges allowed during the 
	construction of the 
	cover. \label{tab:test}}
\end{table}

\begin{table}[t]
	\centering
	\begin{tabular}{rrrrrrrrr}
		  \toprule
		       & \multicolumn{4}{c}{Graphs} & \multicolumn{4}{c}{Generalizations}\\
		       \cmidrule(lr){2-5}
		       \cmidrule(lr){6-9}
		     d & 3,5,d & 4,4,24-d & problems & days & 3,5,d
		       & 4,4,24-d & problems & days\\
		     \midrule
		     8 & 179 & 2 & 358 & 0.055 & 27 & 2 & 54 & 0.018\\
		     10 & 313 & 130816 & 40945408 & 8373 & 43 & 11752 & 505336 & 572\\
		     12 & 12 & 1449166 & 17389992 & 7702 & 12 & 26845 & 322140 & 374\\
		\bottomrule
	\end{tabular}
	\caption{Reduction of the number of SAT solver 
	calls and faster estimated times in days \label{tab:reduction}
 }
\end{table}

\section{Combining the Different Parts of the Proof}\label{sec:final}
Our proof is expressed using three different formal representations of
mathematical statements. 
In Section~\ref{sec:deg}, we express our statements in a higher-order form 
allowing 
us to make counting arguments. In Section~\ref{sec:proofenum}, we rely on an 
almost first-order representation to implement a custom
theorem prover for graphs and in particular to prove isomorphism between graphs. 
In Section~\ref{sec:gluing}, the problems are stated at the propositional 
level. Here, we first describe how we connect the different representations and as 
a consequence prove that $R(4,5) \leq 25$. Then, we give the proof of the existence
of an $\cR(4,5,24)$-graph and show that $R(4,5) > 24$. 

\subsection{Connecting Representations}\label{sec:connect}
We will start by translating propositional gluing lemmas to first-order formulas.
The SAT problems do not explicitly mention on which vertices
the graphs are lying on since they are only constraining SAT variables that
represent edges.
Surprisingly, one can instantiate the SAT variables $E_{i,j}$ by the
atom $E\ x_i\ x_j$ in the gluing problem. As a consequence, gluing problems for 
all permutations of edges are proved at once. 
We can also freely add the following additional constraints. All variables $x_i$ 
must be distinct and variables with indices less than the degree $d$ must have a 
value less than $d$ and other variables must have a value greater or equal to $d$. 
This ensures that our generalizations 
$G^*$ and $H^*$ have distinct sets of vertices.
We then prove that for each $\cR^*(3,5,d)$-generalization a
single theorem stating that this particular generalization can not be glued to
any of the corresponding generalizations in $\cR^*(4,4,24-d)$ by regrouping
gluing theorems. This step constructs $27$ theorems for degree $d=8$, 
$43$ theorems for degree $d=10$, and $12$ theorems for degree $d=12$.
These numbers correspond to the number of $\cR^*(3,5,d)$-generalizations
presented in Table~\ref{tab:cover}.
Together, these 27 theorems (respectively 43 and 12) can be used to prove a theorem 
stating that the splitting edge, represented by the vertex number 24,
cannot have degree 8 (respectively 10 and 12).
The higher-order version of these three final theorems
do not state on which set of vertices the neighbors and antineighbors should
lie although it requires them to form sets of nonnegative integers of
size $d$ and $24-d$ respectively. To prove the more general higher-order 
formulations, we rely on the fact that there is
a bijection from $[|0,d-1|]$ to sets of vertices of size $d$ and
a bijection from $[|d,23|]$ to sets of vertices of size $24-d$.
These three theorems together are enough
to prove that $R(4,5) \leq 25$ according to the proof given in 
Section~\ref{sec:deg}.

\subsection{Existence of an 
\texorpdfstring{$\cR(4,5,24)$}{R(4,5,25)}-graph}~\label{sec:existence}
To prove the existence of an $\cR(4,5,24)$-graph, we pick a graph
from the full list of $\cR(4,5,24)$-graphs compiled in 2016 and available
at~\cite{ramsey-graphs}.
This was necessary step 
to prove that $R(5,5)\leq 48$ as described 
in~\cite{DBLP:journals/jgt/AngeltveitM18}. 
For our purpose, we only need one arbitrary witness graph $G_0$ from that list. 
Let $B$ be the set of blue edges in $G_0$, we represent the graph $G_0$ as the 
relation: 
  \[E_0 =_{\mathit{def}} \lambda i j.\  
  \hspace{-4mm}\bigvee\limits_{(a,b) \in B \wedge a < b}\hspace{-4mm}
    (i  = a \wedge j = b) \vee (j = a \wedge i = b)\]
We prove on the first-order level
that this graph does not contain any blue 4-cliques or any red 5-cliques
which can be stated as:
\begin{align*}
&\vdash \forall\underset{\mathit{distinct}}{x_0 x_1 x_2 x_3} < 24.\
  \neg E_0 x_0 x_1 \vee \neg E_0 x_0 x_2 \vee \neg E_0 x_0 x_3 \vee 
  \neg E_0 x_1 x_2 \vee \neg E_0 x_1 x_3 \vee \neg E_0 x_2 x_3\\
&\vdash \forall\underset{\mathit{distinct}}{x_0 x_1 x_2 x_3 x_4} < 24.\
E_0 x_0 x_1 \vee E_0 x_0 x_2 \vee E_0 x_0 x_3 \vee E_0 x_0 x_4 \vee E_0 x_1 x_2 
\vee \\
&\hspace{33mm} E_0 x_1 x_3 \vee E_0 x_1 x_4 \vee E_0 x_2 x_3
  \vee E_0 x_2 x_4 \vee E_0 x_3 x_4
\end{align*}
This was achieved by repeatedly applying the following lemma to 
eliminate the quantified variables $\vdash (\bigwedge\limits_{x<24} P(x)) 
\Rightarrow (\forall x<24. 
P(x))$. To speed up the process, we stop applying the lemma as 
soon as we were able to prove the goal on the branch either because we find a red 
edge in blue clique (or a blue edge in a red clique) or because we have selected 
the same vertex twice in the clique. With these optimizations, the
existence of a graph can be verified in less than 15 minutes on a single CPU.
The connection with the higher-order formulation can be obtained
by proving that the set $[|0,23|]$ has cardinality 24.
And thus we get $R(4,5) > 24$ which together with $R(4,5) \leq 25$ gives:
 \[\vdash R(4,5)=25\]

\section{Reproducibility}\label{sec:repro}
We provide instructions on how to reproduce the proof in 
the \texttt{README.md} of our repository~\cite{ramsey-github}.

The computational resources necessary to run our proof are the following.
The final gluing step was run on 4 different machines allowing us
to finish the gluing phase in less than 9 days. This is slightly longer
than what we expected according to the estimated times. Two machines
were used for the $d=10$ case and the other 2 were used for the $d=12$ case. 
Three of them have 512 GB of RAM and one of them has 1024 GB of RAM.
All of those machines have 64 hyper-threaded CPU cores for a total of 128 
available CPUs. However, we only used 40 cores per machine as our main limitation
was memory. In all those machines, the same copy of \holfour was used
to get compatible timestamps in the produced gluing theories.
The other phases were much faster and required less memory
but were still run on the machine with more memory.
 
Potential issues one could encounter when trying to reproduce the 
proof are the following. As expected, all the technical problems come during or 
after the expensive 
gluing phase.
The first issue we discovered is that the communication files
between HOL4 and \minisat are stored in the temporary directory of the system.
Since the proof file produced by \minisat can be up to 2GB and our temporary 
directory sits in a partition of only 32GB, we ran out of memory in that partition
because we were running 40 processes in parallel.
So, we changed the temporary directory used by 
the HOL4 interface to \minisat by modifying the file \texttt{dimacsTools.sml}. We 
changed the temporary directory used by \minisat using the \texttt{TMPDIR} 
bash variable.
The second issue is that the reconstruction of a SAT proof in \holfour can
require a lot of memory (in the order of 20GB per problem) and 
time (about 3 times longer than the SAT solver call). 
We found out that creating theories with one theorem per 
theory diminished memory consumption. To guarantee that the memory consumption 
did not exceed a threshold, we also ran the scripts using \texttt{buildheap} 
instead of \texttt{Holmake} as the latter does not provide a way to limit the 
memory consumption per core.
Finally, we were not able to load all the gluing theories together into HOL4 in a 
reasonable time. Indeed, we observed a slowdown of time taken to load one theory 
as more and more theories were added making it impossible for us to load more than 
200,000 theories.
Thus, we used the following workaround instead.
We loaded the theorems produced by the 43 theorems for degree $d=10$ and
the 12 theorems for degree $d=12$, mentioned in Section~\ref{sec:connect},
without specifying in which theory they were proved. 
This was achieved
by creating an alternative theory loader where we omitted the call to 
the function \texttt{link\_parents}.
To ensure the safety of this procedure, we externally check 
that constructing the complete theory graph for our proof without any 
broken dependencies is possible in the directory \texttt{theorygraph} . There, we 
make sure that the time stamps are coherent and that there are no cycles.
In the complete theory graph, the final theory \texttt{r45\_equals\_24} has 
828857 ancestor theories.

\section{Related Works}\label{sec:rel}
Our formalization is based on the work of McKay and 
Radziszowski~\cite{DBLP:journals/jgt/McKayR95}.
Their proof already contains the three steps performed in our formalization, 
namely:
the degree constraints for a splitting, the creation of covers, and the proof
of the absence of satisfiable gluing problems.
We explain here how our proof differs in each of those steps.
In the first step, their proof mentions that it is not possible for all vertices
to have degree 11. We realized during the formalization that
this argument can be applied to prove that it is not possible for all vertices to  
have odd degrees.
This allows us 
to save time during the formalization compared to their original proof by 
additionally ignoring the cases of a splitting vertex of degree $d\in\{7,9,13\}$.
In the second step, their proof removes vertices from graphs to create 
generalizations which would correspond to graying all edges connected to
the removed vertices in our formalization. We tried this approach but the
large number of grayed edges made the problems very difficult for the
SAT solvers. Moreover, our tests gave an estimated time to completion much larger 
than with our more parsimonious approach to graying edges.
In the final step, they rely on a custom provers
for gluing generalizations and
spend a large part of the paper describing how they optimize its different
components. In contrast, we perform the gluing step with an existing SAT solver.
Furthermore, they state that
they prefer sparser (redder) generalizations when constructing a cover of 
$\cR(3,5,10)$-graphs and
denser (bluer) generalizations when constructing a cover of $\cR(4,4,12)$-graphs. 
Our approach instead relies on a more involved heuristic. 
We try to estimate the simplicity of a generalization $G^*$ by understanding 
which type of clauses would appear in a gluing problem 
where one of the generalizations is $G^*$.


Proving mathematical theorems in combinatorics with the help of SAT solvers
is not a new phenomenon.
For instance in~\cite{DBLP:journals/constraints/CodishFIM16}, the authors prove 
using SAT solvers that the 3-color Ramsey number 
R(4,3,3) is equal to 30. 
This proof contains consideration about the degrees of vertices,
a symmetry-breaking component to avoid considering
isomorphic graphs and an abstraction component relying on degree matrices
to represent sets of graphs. Ultimately, the approach relies on 
encoding each of these steps into SAT clauses and calling a SAT solver.
A more famous example is the proof of the Pythagorean Triple 
theorem~\cite{DBLP:conf/sat/HeuleKM16}.
In that paper, the authors rely on a cube-and-conquer, look-ahead methods and 
symmetry-breaking arguments. The DRAT proof produced by their custom SAT solver
was verified using an independent DRAT proof 
checker~\cite{DBLP:conf/sat/WetzlerHH14}.  
In a later work, the proof of the Pythagorean Triple 
theorem has been fully formalized in 
\coq~\cite{DBLP:journals/jar/Cruz-FilipeMS19}. There, they verified
symmetry-breaking arguments and an encoding of the problem into SAT.
For the computational part of the proof, they relied on \textsf{OCaml} 
code proven correct in \coq. This last 
step was necessary because of memory limitations but is generally considered 
slightly less safe than running the computation directly in \coq. In contrast, we 
were able to run the entire proof of $R(4,5)=25$, which is larger in
size (approximately a petabyte of DRAT proof files versus 200 terabytes), through 
the \holfour 
kernel.


\section{Conclusion}~\label{sec:concl}
We have created the first formalization of the theorem $R(4,5)=25$. This
verification was performed within the \holfour theorem prover.
During this process, we have realized that we can generalize the argument
given for degree $d=11$ to eliminate all odd degree cases.
We have designed a verified algorithm for regrouping similar graphs, creating 
generalizations and ultimately speeding up our subsequent proofs.
Finally, we have created and tested a heuristic for predicting the relative 
run-time of a SAT solver on gluing problems. This helped us choose generalizations 
that create easier gluing problems.

In the future, we would like to investigate if a proof of $R(4,5)=25$ is
intrinsically computational or if there exist additional high-level arguments that
could eliminate the need for a large computation.

\bibliographystyle{plainurl}
\bibliography{biblio,ate11}

\begin{thebibliography}{10}

\bibitem{DBLP:journals/jgt/AngeltveitM18}
Vigleik Angeltveit and Brendan~D. McKay.
\newblock R(5, 5) {\(\leq\)} 48.
\newblock {\em J. Graph Theory}, 89(1):5--13, 2018.
\newblock URL: \url{https://doi.org/10.1002/jgt.22235}, \href
  {https://doi.org/10.1002/JGT.22235} {\path{doi:10.1002/JGT.22235}}.

\bibitem{ramsey-graphs}
Vigleik Angeltveit, Brendan~D. McKay, and Stanislaw~P. Radziszowski.
\newblock Website section: {A}ll maximal {R}amsey(4,5)-graphs.
\newblock \url{https://users.cecs.anu.edu.au/~bdm/data/ramsey.html}, 2016.
\newblock Accessed on May 26, 2024.

\bibitem{appel1989every}
Kenneth~I Appel and Wolfgang Haken.
\newblock {\em Every planar map is four colorable}, volume~98.
\newblock American Mathematical Soc., 1989.

\bibitem{coq}
Yves Bertot.
\newblock A short presentation of {C}oq.
\newblock In Otmane~Ait Mohamed, César Muñoz, and Sofiène Tahar, editors,
  {\em Conference on Theorem Proving in Higher Order Logics (TPHOLs)}, volume
  5170 of {\em LNCS}, pages 12--16. Springer, 2008.
\newblock URL: \url{http://doi.org/10.1007/978-3-540-71067-7_3}.

\bibitem{DBLP:books/daglib/0009415}
B{\'{e}}la Bollob{\'{a}}s.
\newblock {\em Modern Graph Theory}, volume 184 of {\em Graduate Texts in
  Mathematics}.
\newblock Springer, 2002.
\newblock \href {https://doi.org/10.1007/978-1-4612-0619-4}
  {\path{doi:10.1007/978-1-4612-0619-4}}.

\bibitem{DBLP:journals/constraints/CodishFIM16}
Michael Codish, Michael Frank, Avraham Itzhakov, and Alice Miller.
\newblock Computing the {Ramsey} number {R}(4, 3, 3) using abstraction and
  symmetry breaking.
\newblock {\em Constraints An Int. J.}, 21(3):375--393, 2016.
\newblock URL: \url{https://doi.org/10.1007/s10601-016-9240-3}, \href
  {https://doi.org/10.1007/S10601-016-9240-3}
  {\path{doi:10.1007/S10601-016-9240-3}}.

\bibitem{DBLP:journals/jar/Cruz-FilipeMS19}
Lu{\'{\i}}s Cruz{-}Filipe, Jo{\~{a}}o Marques{-}Silva, and Peter
  Schneider{-}Kamp.
\newblock Formally verifying the solution to the {B}oolean {P}ythagorean
  triples problem.
\newblock {\em J. Autom. Reason.}, 63(3):695--722, 2019.
\newblock URL: \url{https://doi.org/10.1007/s10817-018-9490-4}, \href
  {https://doi.org/10.1007/S10817-018-9490-4}
  {\path{doi:10.1007/S10817-018-9490-4}}.

\bibitem{DBLP:journals/jacm/DavisP60}
Martin Davis and Hilary Putnam.
\newblock A computing procedure for quantification theory.
\newblock {\em J. {ACM}}, 7(3):201--215, 1960.
\newblock \href {https://doi.org/10.1145/321033.321034}
  {\path{doi:10.1145/321033.321034}}.

\bibitem{DBLP:conf/sat/EenS03}
Niklas E{\'{e}}n and Niklas S{\"{o}}rensson.
\newblock An extensible {SAT}-solver.
\newblock In Enrico Giunchiglia and Armando Tacchella, editors, {\em Theory and
  Applications of Satisfiability Testing, 6th International Conference, {SAT}
  2003. Santa Margherita Ligure, Italy, May 5-8, 2003 Selected Revised Papers},
  volume 2919 of {\em Lecture Notes in Computer Science}, pages 502--518.
  Springer, 2003.
\newblock \href {https://doi.org/10.1007/978-3-540-24605-3\_37}
  {\path{doi:10.1007/978-3-540-24605-3\_37}}.

\bibitem{ramsey-github}
Thibault Gauthier and Chad~E. Brown.
\newblock Software accompanying the paper "{A} formal proof of {R}(4,5)=25".
\newblock \url{https://github.com/barakeel/ramsey}, 2024.
\newblock Accessed on May 26, 2024.

\bibitem{gonthier08ams}
Georges Gonthier.
\newblock Formal proof the four-color theorem.
\newblock {\em Notices of the AMS}, 55(11):1382--1393, 2008.
\newblock URL: \url{http://www.ams.org/notices/200811/tx081101382p.pdf}.

\bibitem{DBLP:journals/corr/HalesABDHHKMMNNNOPRSTTTUVZ15}
Thomas~C. Hales, Mark Adams, Gertrud Bauer, Dat~Tat Dang, John Harrison,
  Truong~Le Hoang, Cezary Kaliszyk, Victor Magron, Sean McLaughlin, Thang~Tat
  Nguyen, Truong~Quang Nguyen, Tobias Nipkow, Steven Obua, Joseph Pleso,
  Jason~M. Rute, Alexey Solovyev, An~Hoai~Thi Ta, Trung~Nam Tran, Diep~Thi
  Trieu, Josef Urban, Ky~Khac Vu, and Roland Zumkeller.
\newblock A formal proof of the {K}epler conjecture.
\newblock {\em CoRR}, abs/1501.02155, 2015.
\newblock URL: \url{http://arxiv.org/abs/1501.02155}, \href
  {http://arxiv.org/abs/1501.02155} {\path{arXiv:1501.02155}}.

\bibitem{DBLP:journals/dcg/HalesF06}
Thomas~C. Hales and Samuel~P. Ferguson.
\newblock A formulation of the {K}epler conjecture.
\newblock {\em Discret. Comput. Geom.}, 36(1):21--69, 2006.
\newblock URL: \url{https://doi.org/10.1007/s00454-005-1211-1}, \href
  {https://doi.org/10.1007/S00454-005-1211-1}
  {\path{doi:10.1007/S00454-005-1211-1}}.

\bibitem{harrison09hollight}
John Harrison.
\newblock {HOL Light}: An overview.
\newblock In Stefan Berghofer, Tobias Nipkow, Christian Urban, and Makarius
  Wenzel, editors, {\em Conference on Theorem Proving in Higher Order Logics
  (TPHOLs)}, volume 5674 of {\em LNCS}, pages 60--66. Springer, 2009.
\newblock URL: \url{http://dx.doi.org/10.1007/978-3-642-03359-9_4}.

\bibitem{DBLP:conf/sat/HeuleKM16}
Marijn J.~H. Heule, Oliver Kullmann, and Victor~W. Marek.
\newblock Solving and verifying the boolean pythagorean triples problem via
  cube-and-conquer.
\newblock In Nadia Creignou and Daniel~Le Berre, editors, {\em Theory and
  Applications of Satisfiability Testing - {SAT} 2016 - 19th International
  Conference, Bordeaux, France, July 5-8, 2016, Proceedings}, volume 9710 of
  {\em Lecture Notes in Computer Science}, pages 228--245. Springer, 2016.
\newblock \href {https://doi.org/10.1007/978-3-319-40970-2\_15}
  {\path{doi:10.1007/978-3-319-40970-2\_15}}.

\bibitem{kalbfleisch1965construction}
JG~Kalbfleisch.
\newblock Construction of special edge-chromatic graphs.
\newblock {\em Canadian Mathematical Bulletin}, 8(5):575--584, 1965.

\bibitem{DBLP:journals/jsc/McKayP14}
Brendan~D. McKay and Adolfo Piperno.
\newblock Practical graph isomorphism, {II}.
\newblock {\em J. Symb. Comput.}, 60:94--112, 2014.
\newblock URL: \url{https://doi.org/10.1016/j.jsc.2013.09.003}, \href
  {https://doi.org/10.1016/J.JSC.2013.09.003}
  {\path{doi:10.1016/J.JSC.2013.09.003}}.

\bibitem{DBLP:journals/jgt/McKayR95}
Brendan~D. McKay and Stanislaw~P. Radziszowski.
\newblock \emph{R}(4, 5) = 25.
\newblock {\em J. Graph Theory}, 19(3):309--322, 1995.
\newblock URL: \url{https://doi.org/10.1002/jgt.3190190304}, \href
  {https://doi.org/10.1002/JGT.3190190304} {\path{doi:10.1002/JGT.3190190304}}.

\bibitem{hol4}
Konrad Slind and Michael Norrish.
\newblock A brief overview of {HOL4}.
\newblock In Otmane~A\"{\i}t Mohamed, C{\'e}sar~A. Mu{\~n}oz, and Sofi{\`e}ne
  Tahar, editors, {\em Conference on Theorem Proving in Higher Order Logics
  (TPHOLs)}, volume 5170 of {\em LNCS}, pages 28--32. Springer, 2008.
\newblock URL: \url{http://dx.doi.org/10.1007/978-3-540-71067-7_6}.

\bibitem{DBLP:journals/japll/WeberA09}
Tjark Weber and Hasan Amjad.
\newblock Efficiently checking propositional refutations in {HOL} theorem
  provers.
\newblock {\em J. Appl. Log.}, 7(1):26--40, 2009.
\newblock URL: \url{https://doi.org/10.1016/j.jal.2007.07.003}, \href
  {https://doi.org/10.1016/J.JAL.2007.07.003}
  {\path{doi:10.1016/J.JAL.2007.07.003}}.

\bibitem{DBLP:conf/sat/WetzlerHH14}
Nathan Wetzler, Marijn Heule, and Warren A.~Hunt Jr.
\newblock {DRAT}-trim: Efficient checking and trimming using expressive clausal
  proofs.
\newblock In Carsten Sinz and Uwe Egly, editors, {\em Theory and Applications
  of Satisfiability Testing - {SAT} 2014 - 17th International Conference, Held
  as Part of the Vienna Summer of Logic, {VSL} 2014, Vienna, Austria, July
  14-17, 2014. Proceedings}, volume 8561 of {\em Lecture Notes in Computer
  Science}, pages 422--429. Springer, 2014.
\newblock \href {https://doi.org/10.1007/978-3-319-09284-3\_31}
  {\path{doi:10.1007/978-3-319-09284-3\_31}}.

\end{thebibliography}
\end{document}